\documentclass[a4paper,twoside]{article}      
\usepackage{amsmath,amssymb,amsfonts,amsthm,amscd,mathtools}
\usepackage{graphics}                 
\usepackage{color}                    
\usepackage{hyperref}                
\usepackage{ mathrsfs }
\usepackage{indentfirst}
\usepackage{bbold}
\usepackage{bm}
\usepackage{enumerate}
\usepackage{url}         
\usepackage{colonequals} 
\usepackage{authblk} 
\usepackage{a4wide}
\usepackage{graphicx}
\usepackage{mdframed}
\mdfsetup{frametitlealignment=\center}
\hypersetup{
    bookmarks=true,         
    unicode=false,          
    pdftoolbar=true,        
    pdfmenubar=true,        
    pdffitwindow=false,     
    pdfstartview={FitH},    
    pdftitle={My title},    
    pdfauthor={Author},     
    pdfsubject={Subject},   
    pdfcreator={Creator},   
    pdfproducer={Producer}, 
    pdfkeywords={keywords}, 
    pdfnewwindow=true,      
    colorlinks=true,       
    linkcolor=blue,          
    citecolor=blue,        
    filecolor=magenta,      
    urlcolor=cyan           
}

\newtheorem{theorem}{Theorem}[section]
\newtheorem{proposition}[theorem]{Proposition}

\newtheorem{lemma}[theorem]{Lemma}
\theoremstyle{definition}
\newtheorem{remark}[theorem]{Remark}

\def\!{\mathop{\mathrm{!}}}

\def\x{\textbf{x}}

\def\f{\bar{f}}

\def\hT{\hat{T}}

\newlength{\boxwidth}
\author[1]{Manh Hong Duong}
\author[2]{The Anh Han}
\affil[1]{School of Mathematics, University of Birmingham, Birmingham B15 2TT, UK. Email: h.duong@bham.ac.uk}
\affil[2]{School of Computing and Digital Technologies, Teesside University, TS1 3BX, UK. Email: T.Han@tees.ac.uk}
\title{Statistics of the number of equilibria in random social dilemma evolutionary games with mutation}
\begin{document}
\maketitle
\begin{abstract}
In this paper, we study analytically the statistics of the number of equilibria in pairwise social dilemma evolutionary games with mutation where a game's payoff entries are random variables. Using the replicator-mutator equations, we provide explicit formulas for the probability distributions of the number of equilibria as well as other statistical quantities. This analysis is highly relevant assuming that one might know the nature of a  social dilemma game at hand (e.g., cooperation vs coordination vs anti-coordination), but measuring the exact  values of its payoff entries is difficult. Our delicate analysis shows clearly the influence of the  mutation probability on these probability distributions, providing insights into how varying this important factor impacts the overall  behavioural or biological diversity of the underlying evolutionary systems.
\end{abstract}
\section{Introduction}
\subsection{The replicator-mutator equation}
The replicator-mutator equation is a set of differential equations describing the evolution of frequencies of different strategies in a population that takes into account both selection and mutation mechanisms. It has been employed in the study of, among other applications, population genetics \cite{Hadeler1981}, autocatalytic reaction networks
\cite{StadlerSchuster1992}, language evolution
\cite{Nowaketal2001}, the evolution of cooperation \cite{Imhof-etal2005,nowak:2006bo} and dynamics of behavior in social networks \cite{Olfati2007}. 

Suppose that in an infinite population there are $n$ types/strategies $S_1,\cdots, S_n$ whose frequencies are, respectively, $x_1,\cdots, x_n$. These types undergo selection; that is, the reproduction rate of each type, $S_i$, is determined by its fitness or average payoff, $f_i$, which is obtained from interacting with other individuals in the population. The interaction of the individuals in the population is carried out within randomly selected groups of $d$ participants (for some integer $d$). That is, they play and obtain their payoffs from a $d$-player game, defined by a payoff matrix. We consider here symmetric games where the payoffs do not depend on the ordering of the players in a group. Mutation is included by adding the possibility that individuals spontaneously change from one strategy to another, which is modeled via a mutation matrix, $Q=(q_{ji}), j,i\in\{1,\cdots,n\}$. The entry $q_{ji}$ denotes the  probability that a player of type $S_j$ changes its type or strategy to $S_i$. The mutation matrix $Q$ is a row-stochastic matrix, i.e.,
\[
\sum_{j=1}^n q_{ji}=1, \quad 1\leq i\leq n.
\]
The replicator-mutator is then given by, see e.g. \cite{Komarova2001JTB,Komarova2004, Komarova2010, Pais2012} 
\begin{equation}
\label{eq: RME}
\dot{x}_i=\sum_{j=1}^n x_j f_j(\x)q_{ji}- x_i \f(\x)=:g_i(x),\qquad  i=1,\ldots, n,
\end{equation}
where $\x = (x_1, x_2, \dots, x_n)$ and $\f(\x)=\sum_{i=1}^n x_i f_i(\x)$ denotes the average fitness of the whole population.  The replicator dynamics is a special instance of \eqref{eq: RME} when the mutation matrix is the identity matrix. 

\subsection{The replicator-mutator equation for two-player two-strategy games}
In particular, for two-player two-strategy games the replicator-mutator equation is
\begin{multline}
\label{eq: 2-2 games 1}
\dot{x}=q_{11}a_{11}x^2+q_{11}x(1-x)a_{12}+q_{21}x(1-x)a_{21}+q_{21}a_{22}(1-x)^2\\-x\Big(a_{11}x^2+(a_{12}+a_{21})x(1-x)+a_{22}(1-x)^2\Big),
\end{multline}
where $x$ is the frequency of the first strategy and $1-x$ is the frequency of the second one.
Using the identities $q_{11}=q_{22}=1-q, \quad q_{12}=q_{21}=q$, Equation \eqref{eq: 2-2 games 1} becomes
\begin{align}
\dot{x}&=\Big(a_{12}+a_{21}-a_{11}-a_{22}\Big)x^3+\Big(a_{11}-a_{21}-2(a_{12}-a_{22})+q(a_{22}+a_{12}-a_{11}-a_{21})\Big)x^2\nonumber
\\&\quad+\Big(a_{12}-a_{22}+q(a_{21}-a_{12}-2a_{22})\Big)x+q a_{22}.\label{eq: 2-2 games}
\end{align}
\paragraph{Two-player social dilemma games.}
In this paper, we focus on two-player (i.e. pairwise) social dilemma games. We adopt the following parameterized payoff matrix to study the full space of two-player  social dilemma games where the first strategy is cooperator and second is defector \cite{santos:2006pn,wang2015universal,szolnoki2019seasonal}, $a_{11} = 1; \  a_{22} = 0; $
$0 \leq a_{21} = T \leq 2$ and $-1 \leq a_{12} = S \leq 1$, that covers the following games
\begin{enumerate}[(i)]
\item the Prisoner's Dilemma (PD): $2\geq T > 1 > 0 > S\geq -1$,
\item the Snow-Drift (SD) game: $2\geq T > 1 > S > 0$,
\item the Stag Hunt (SH) game: $1 > T  > 0 > S\geq -1$,
\item the Harmony (H) game: $1  > T\geq 0, 1\geq S > 0$.
\end{enumerate}
In this paper, we are interested in random social dilemma games where $T$ and $S$ are uniform random variables in the corresponding intervals, namely
\begin{itemize}
\item In PD games: $T\sim U(1,2)$, $S\sim U(-1,0)$,
\item In SD games: $T\sim U(1,2)$, $S\sim U(0,1)$,
\item In SH games: $T\sim U(0,1)$, $S\sim U(-1,0)$,
\item In H games: $T\sim U(0,1)$, $S\sim U(0,1)$.
\end{itemize} 
Random evolutionary games, in which the payoff entries are random variables, have been employed extensively to model social and biological systems in which  very limited information is available, or where the environment changes so rapidly and frequently that one cannot describe the payoffs of their inhabitants' interactions  \cite{broom:1993pa,broom:2005aa,may2001stability,HTG12,DH15,DuongHanJMB2016,gokhale:2010pn,DuongTranHanJMB}. 
Equilibrium points of such evolutionary system are the compositions of strategy frequencies where all the strategies have the same average fitness. Biologically, they predict the co-existence of different types in a population and the maintenance of polymorphism.

In this paper, we are interested in computing the probability distributions of the number of equilibria, which is a random variable, in the above random social dilemmas. Answering this question is of great importance  in the context of social dilemmas since  one  might  know the nature of the game, i.e. the payoff entries ranking in the game, but it might be difficult to predict or measure the exact values  of these entries. When mutation is absent ($q = 0$),  the answer is trivial because there is always a fixed number of equilibria depending on the nature of the social dilemmas \cite{nowak:2006bo,sigmund:2010bo,DuongHanDGA2020}. As shown in our analysis below, this is however not the case any longer when mutation is non-negligible, and this number  highly depends on the nature of the social dilemma too.  

The following result \cite{DuongHanDGA2020} provides partial information about these probability distributions.
\begin{theorem}({DuongHanDGA2020})
\label{thm: previous}
Suppose that $S$ and $T$ are uniformly distributed in the corresponding intervals as above. Then
\begin{itemize}
\item $p_1^{SD}=p_3^{SD}=0, \quad p_2^{SD}=1$.
\item $p_1^{H}=p_3^{H}=0,\quad p_2^{H}=1$.
\item $p_2^{SH}=\frac{q}{2(1-q)}$.
\item $p_2^{PD}=\begin{cases}
\frac{3 q}{2(1-q)}\quad \text{if}\quad 0<q\leq 1/3,\\
3-\frac{1}{2q(1-q)}\quad \text{if}\quad 1/3\leq q\leq 1/2. 
\end{cases}$
\end{itemize}
\end{theorem}
According to the above theorem, SD games and H games are simple and complete. However, the probabilities of having $3$ equilibria (or alternatively $1$ equilibrium) in  SH and PD games are left open in \cite{DuongHanDGA2020}. The key challenge is that the conditions for these games to have 3 equilibria (or alternatively 1 equilibrium) are much more complicated than those of  2 equilibria. The aim of this paper is to complete the above theorem, providing explicit formulas for these probabilities. As such, it will also allow us to derive other statistical quantities (such as average and variance), which are important to understand the overall distribution and complexity of equilibrium points in pairwise social dilemmas (with mutation). To this goal, we employ suitable changes of variables, which transform the problem of computing the probabilities to calculating areas, and perform delicate analysis.
\subsection{Main results}
The main result of this paper is the following.
\begin{theorem}
\label{thm: main thm}
The probability that SH and PD games have $3$ equilibria is given by, respectively
\begin{align*}
&p_3^{SH}=
\begin{cases}
1-\frac{q}{2(1-q)}-\frac{1}{1-2q}\Big[\frac{(3\sqrt{q}+2)^2\sqrt{q}(5q^{3/2}+3q^2-9q-3\sqrt{q}+4)}{12(\sqrt{q}+1)^3}
\\ \hspace*{5cm}+~\frac{-27 q^3-18q^2-32\sqrt{1-2q}q+48q+16\sqrt{1-2q}-16}{12q}\Big],\quad 0<q\leq 4/9\\
1-\frac{q}{2(1-q)}-\frac{8\sqrt{q}(1-2q)^2}{3(1-q)^3},\quad 4/9<q<0.5
\end{cases}
\\&
\\&
\\&p_3^{PD}=\begin{cases}
\frac{1}{1-2q}\Big[-\frac{2(q^3+3q^2+(4\sqrt{1-2q}-6)q-2\sqrt{1-2q}+2)}{3q}-\frac{1}{2}\frac{q^3}{(1-q)}\Big],\quad 0\leq q\leq \frac{3-\sqrt{5}}{2},\\
\frac{-16 \left(\sqrt{1-2 q}-1\right) q^{3/2}+2 q^{5/2}+15 q^3+\left(8 \sqrt{1-2 q}-25\right) q^2+q-8 \sqrt{1-2 q}+2 \sqrt{q} \left(8 \sqrt{1-2 q}-5\right)+5}{6 \left(\sqrt{q}-1\right)^3 \left(q^{3/2}+q\right)},\quad \frac{3-\sqrt{5}}{2}<q\leq 4/9,\\
\frac{1}{2}\frac{(1-2q)^2}{q(1-q)},\quad 4/9<q<0.5.
\end{cases}
\end{align*}
\end{theorem}
The above theorem combines Theorem \ref{thm: SH} (for SH-games) and Theorem \ref{thm: PD} (for PD games), see Section \ref{sec: main}. Theorems \ref{thm: previous} and \ref{thm: main thm} provides explicitly the probability distributions of the number of equilibria for all the above-mentioned pairwise social dilemmas. In SH-games and PD-games, these distributions are much more complicated and significantly depend on the mutation strength. We summarize these results in the following summary box.

\newpage

\begin{mdframed}[linecolor=black,font={\sffamily},frametitle={\large{\underline{Box 1: Probability of having $k$ equilibria in a pairwise social dilemma ($p_k$)}}}]
\vspace{0.5cm}
$\bullet$ \textbf{Snow Drift (SD) }
$$
p_1=0, \quad p_2=1, \quad p_3=0.
$$
$\bullet$ \textbf{Harmony game (H)}
$$
p_1=0, \quad p_2=1, \quad p_3=0.
$$
$\bullet$ \textbf{Stag-Hunt game  (SH)}
\begin{align*}
p_1&=1-p_2-p_3,\\
p_2&=\frac{q}{2(1-q)}\\
p_3&=\begin{cases}
1-\frac{q}{2(1-q)}-\frac{1}{1-2q}\Big[\frac{(3\sqrt{q}+2)^2\sqrt{q}(5q^{3/2}+3q^2-9q-3\sqrt{q}+4)}{12(\sqrt{q}+1)^3}\\\hspace{5cm}+~\frac{-27 q^3-18q^2-32\sqrt{1-2q}q+48q+16\sqrt{1-2q}-16}{12q}\Big],\quad0<q\leq 4/9\\
1-\frac{q}{2(1-q)}-\frac{8\sqrt{q}(1-2q)^2}{3(1-q)^3},\quad 4/9<q<0.5.
\end{cases}
\end{align*}
$\bullet$ \textbf{Prisoner's Dilemma (PD)}
\begin{align*}
p_1&=1-p_2-p_3\\
p_2&=\begin{cases}
\frac{3 q}{2(1-q)}\quad \text{if}\quad 0<q\leq 1/3,\\
3-\frac{1}{2q(1-q)}\quad \text{if}\quad 1/3\leq q\leq 1/2. 
\end{cases}\\
p_3&=\begin{cases}
\frac{1}{1-2q}\Big[-\frac{2(q^3+3q^2+(4\sqrt{1-2q}-6)q-2\sqrt{1-2q}+2)}{3q}-\frac{1}{2}\frac{q^3}{(1-q)}\Big],\quad 0\leq q\leq \frac{3-\sqrt{5}}{2},\\
\frac{-16 \left(\sqrt{1-2 q}-1\right) q^{3/2}+2 q^{5/2}+15 q^3+\left(8 \sqrt{1-2 q}-25\right) q^2+q-8 \sqrt{1-2 q}+2 \sqrt{q} \left(8 \sqrt{1-2 q}-5\right)+5}{6 \left(\sqrt{q}-1\right)^3 \left(q^{3/2}+q\right)},\quad \frac{3-\sqrt{5}}{2}<q\leq 4/9,\\
\frac{1}{2}\frac{(1-2q)^2}{q(1-q)},\quad 4/9<q<0.5.
\end{cases}
\end{align*}
\end{mdframed}
As a consequence, we can now derive other statistical quantities such as the mean value, $\mathrm{ENoE}$, and the variance, $\mathrm{VarNoE}$ of the number of equilibria using the following formulas
\begin{equation}
\label{mean and variance}
\mathrm{ENoE}=\sum_{i=1}^3 i\, p_i,\quad \mathrm{VarNoE}=\sum_{i=1}^3 p_i (i-\mathrm{ENoE})^2.
\end{equation}
We depict  these quantities in Figure \ref{fig: plot all games} below. Our delicate analysis clearly shows the influence of the mutation on the probability distributions, thus on the complexity and bio-diversity of the underlying evolutionary systems. We believe that our analysis may be used as exemplary material for  teaching foundational courses in evolutionary game theory, computational/quantitative biology and applied probability.


\begin{figure}[htb!]
\begin{center}
\includegraphics[width=\linewidth]{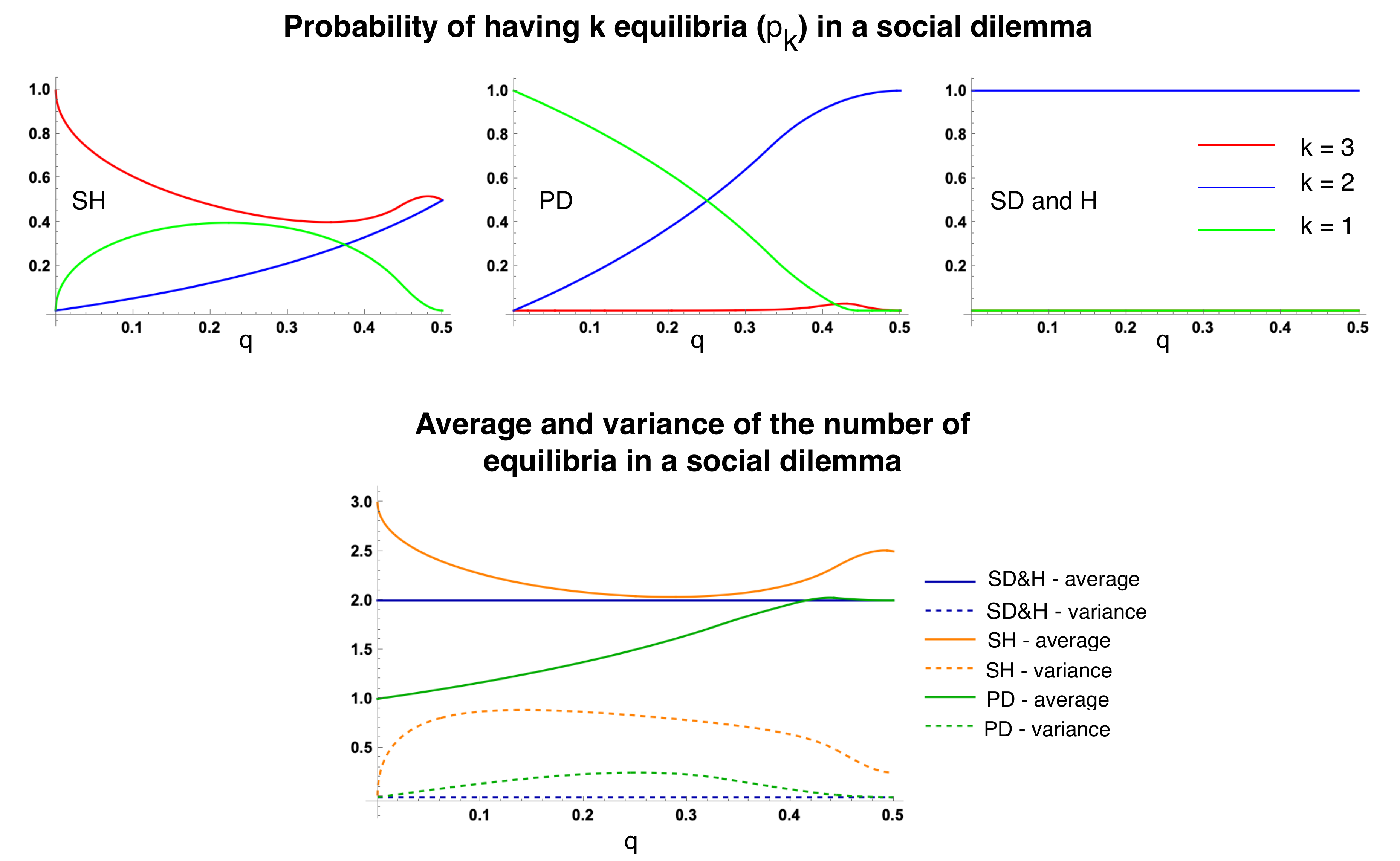}
\caption{Statistics of the number of equilibria in random social dilemmas as a function of the mutation probability $q$: probability of having $k$ ($k = 1, \ 2, \ 3$) equilibria (\textbf{top row}), average and variance (\textbf{bottom row}). We look at all the four pairwise games, from left to right: SH, PD, SH and H. The probability of having the maximal possible number of equilibria, i.e. $p_3$, is highest for SH. It is very small for PD and always equals 0 for SD and H games. The probability of having two equilibria, $p_2$, is highest for SD and H.  As a result, SH has the highest average number of equilibria across all games, for all $0 < q < 0.5$ (it is minimum when $q = 0.285$).  For most $q$, PD has the lowest number of equilibria (it is maximum when $q = 0.438$). All the figures are generated using analytical formulas derived in Box 1 (for $p_k$) and \eqref{mean and variance} for the average and variance. These analytical results are also in accordance with numerical simulation results provided in \cite{DuongHanDGA2020}, obtained  through samplings of the random payoff entries $T$ and $S$.}  
\label{fig: plot all games}
\end{center}
\end{figure}

The rest of the paper is organized as follows. In Section \ref{sec: main}, after recalling some preliminary details, we present the proof of the main theorem \ref{thm: main thm}, which we split into Theorem \ref{thm: SH} for SH games in subsection \ref{sec: SH}  and   Theorem \ref{thm: PD} for PD games in subsection \ref{sec: PD}. Finally, in Section \ref{sec: summary} we provide further discussion.
\section{Probability of having three equilibria in SH  and PD games}
\label{sec: main}
\subsection{Joint probability density via change of variable}
The following lemma is a well-known result to compute the probability density of random variables using change of variables.  We state here for two random variables that are directly applicable to our analysis, but the result is true in higher dimensional spaces.  
\begin{lemma} (joint probability density via change of variable, \cite[Section 3.3]{durrett1994})
\label{lem: joint pdf}
Suppose $(X_1, X_2)$ has joind density $f(x_1,x_2)$. Let $(Y_1,Y_2)$ be defined by $Y_1=u_1(X_1,X_2)$ and $Y_2=u_2(X_1, X_2)$. Suppose that the map $(X_1, X_2)\rightarrow (Y_1, Y_2)$ is invertible with $X_1=v_1(Y_1, Y_2)$ and $X_2=v_2(Y_1, Y_2)$.  Then the joint probability distribution of $Y_1$ and $Y_2$ is
$$
g(y_1, y_2)=|J| f(v_1(y_1,y_2), v_2(y_1, y_2)),
$$
where the Jacobian $J$ is given by
$$
J=\begin{vmatrix}\frac{\partial v_1(y_1, y_2)}{\partial y_1}& \frac{\partial v_1(y_1, y_2)}{\partial y_2}\\ \frac{\partial v_2(y_1,y_2)}{\partial y_1}& \frac{\partial v_2(y_1, y_2)}{\partial y_2}\end{vmatrix}.
$$
\end{lemma}
\subsection{Equilibria in Social dilemmas}
By simplifying the right hand side of \eqref{eq: 2-2 games}, equilibria of a social dilemma game are roots in the interval $[
0,1]$ of the following cubic equation 
\begin{align}
\Big(T+S-1\Big)x^3+\Big(1-T-2S+q(S-1-T)\Big)x^2+\Big(S+q(T-S)\Big)x = 0.\label{eq: 2-2 games2}
\end{align}
It follows that $x = 0$ is always an equilibrium. If $q=0$, \eqref{eq: 2-2 games2} reduces to
$$
(T+S-1)x^3+(1-T-2S)x^2+Sx=0,
$$
which has solutions
$$
x=0, \quad x=1, \quad x^*=\frac{S}{S+T-1}.
$$
Note that for SH-games and SH-games $x^*\in (0,1)$, thus it is always an equilibrium. On the other hand, for PD-games and H-games, $x^* \not\in (0,1)$, thus it is not an equilibrium

If $q=\frac{1}{2}$ then the above equation has two solutions $x_1=\frac{1}{2}$ and $x_2=\frac{T+S}{T+S-1}$. In PD, SD and H games, $x_2\not \in (0,1)$, thus they have two equilibria $x_0=0$ and $x_1=\frac{1}{2}$. In the SH game: if $T+S<0$ then the game has three equilibria $x_0=0, x_1=\frac{1}{2}$ and $0<x_2<1$; if $T+S\geq 0$ then the game has only two equilibria $x_0=0, x_1=\frac{1}{2}$.

Now we consider the case $0<q<\frac{1}{2}$. For non-zero equilibrium points we solve the following quadratic equation
\begin{equation}
\label{eq: social dilemma}
h(x):=(T+S-1)x^2+(1-T-2S+q(S-1-T))x+S+q(T-S)=:ax^2+bx+c=0.
\end{equation}
Set $t:=\frac{x}{1-x}$, then we have 
$$
\frac{h(x)}{(1-x)^2}=(a+b+c) t^2+(b+2c)t+c=-qt^2+(-q-a+c)t+c:=g(t)
$$
Therefore, a social dilemma has three equilibria  iff $h$ has two distinct roots in $(0,1)$, which is equivalent to  $g$ having  two distinct positive roots, or the following conditions must hold 
\begin{equation}
\label{eq: condition}
\Delta=(q+a-c)^2+4qc>0,\quad q+a-c<0,\quad c<0.
\end{equation}
Let $\hat{T}:=T-1$ and
$$
X(\hT, S):=q+a-c=(1-q)(T-1)+qS=(1-q)\hat{T}+qS,\quad Y(\hT,S) :=c-q=q(T-1)+(1-q)=q\hat{T}+(1-q)S.
$$
The inverse transformation $(X,Y)\rightarrow (\hT,S)$ is given by
\begin{equation}
\hT(X,Y)=\frac{(1-q)X-qY}{1-2q},\quad S(X,Y)=\frac{(1-q)Y-qX}{1-2q}.
\end{equation}
Condition \eqref{eq: condition} is given by
$$
D=\{(X,Y): X^2+4qY>-4q^2,\quad X<0, \quad Y<-q\}=\{(X,Y): X<0, -\frac{X^2}{4q}-q<Y<-q\}.
$$
The domain $D$ is illustrated in Figure \ref{fig: D}.
\begin{figure}
\begin{center}
\includegraphics[width=0.5\linewidth]{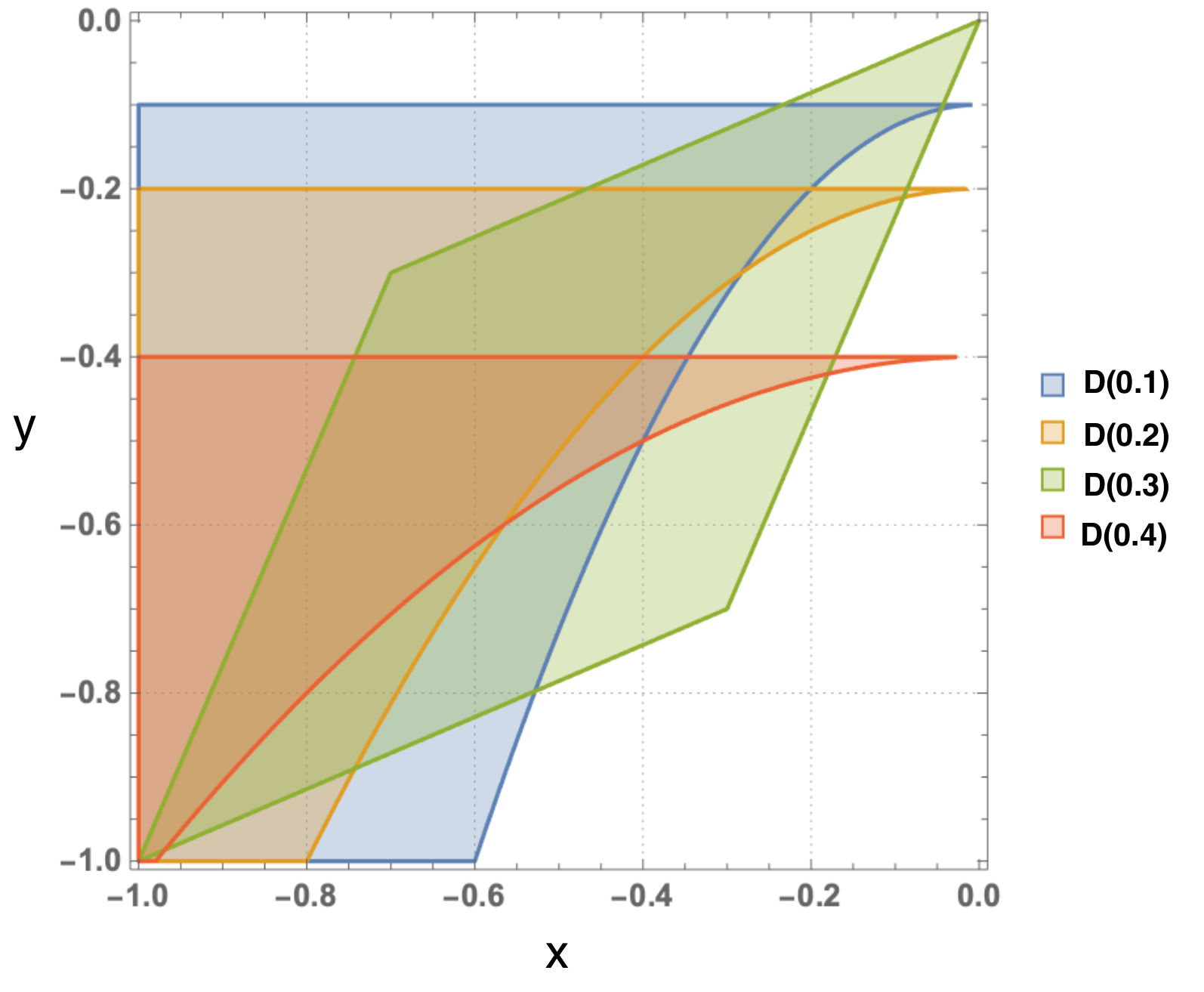}
\caption{Region $D$ are shown for different values of $q$, namely, $q = 0.1, \ 0.2, \ 0.3$ and 0.4.  }
\label{fig: D}
\end{center}
\end{figure}
We apply Lemma \ref{lem: joint pdf} to find the joint distribution of $X$ and $Y$. We compute the Jacobian of the transform $(X,Y)\rightarrow (\hat{T},S)$, which is given by
\begin{equation}
J=\begin{pmatrix}
\frac{\partial \hT(X,Y)}{\partial X}&\frac{\partial S(X,Y)}{\partial Y}\\
\frac{\partial \hT(X,Y)}{\partial X}&\frac{\partial S(X,Y)}{\partial Y}
\end{pmatrix}=\begin{pmatrix}
\frac{1-q}{1-2q}& \frac{-q}{1-2q}\\
\frac{-q}{1-2q}&\frac{1-q}{1-2q}
\end{pmatrix}=\frac{1}{1-2q}.
\end{equation}
Hence if $(\hT, S)$ has a probability density $f(t,s)$ then $(X,Y)$ has a probability density
\begin{equation}
\label{eq: joint pdf of XY}
g(x,y)=|J|f(\hT(x,y),S(x,y)).
\end{equation}
We now apply this approach to SH and PD games.
\subsection{The Stag Hunt (SH) }
\label{sec: SH}
\begin{proposition}
\label{prop: p3areaSH}
The probability that SH games have $3$ equilibria is given by
\begin{equation}
p^{SH}_3=\frac{1}{(1-2q)}\mathrm{Area}(D\cap D_1),
\end{equation}
where $D$ is the subset of $\mathbb{R}^2$ determined by
$$
D=\{(x,y): -1<x<0, -\frac{x^2}{4q}-q<y<-q\},
$$
and $D_1$ is the quadrilateral $ABOC$ with vertices
$$
A=(-1,-1), \quad B=(-(1-q), -q), \quad C=(-q,-(1-q))\quad O=(0,0).
$$
\end{proposition}
The domain $D_1$ and the intersection $D\cap D_1$ is illustrated in Figure \ref{fig: D1}.
\begin{figure}
\begin{center}
\includegraphics[width=\linewidth]{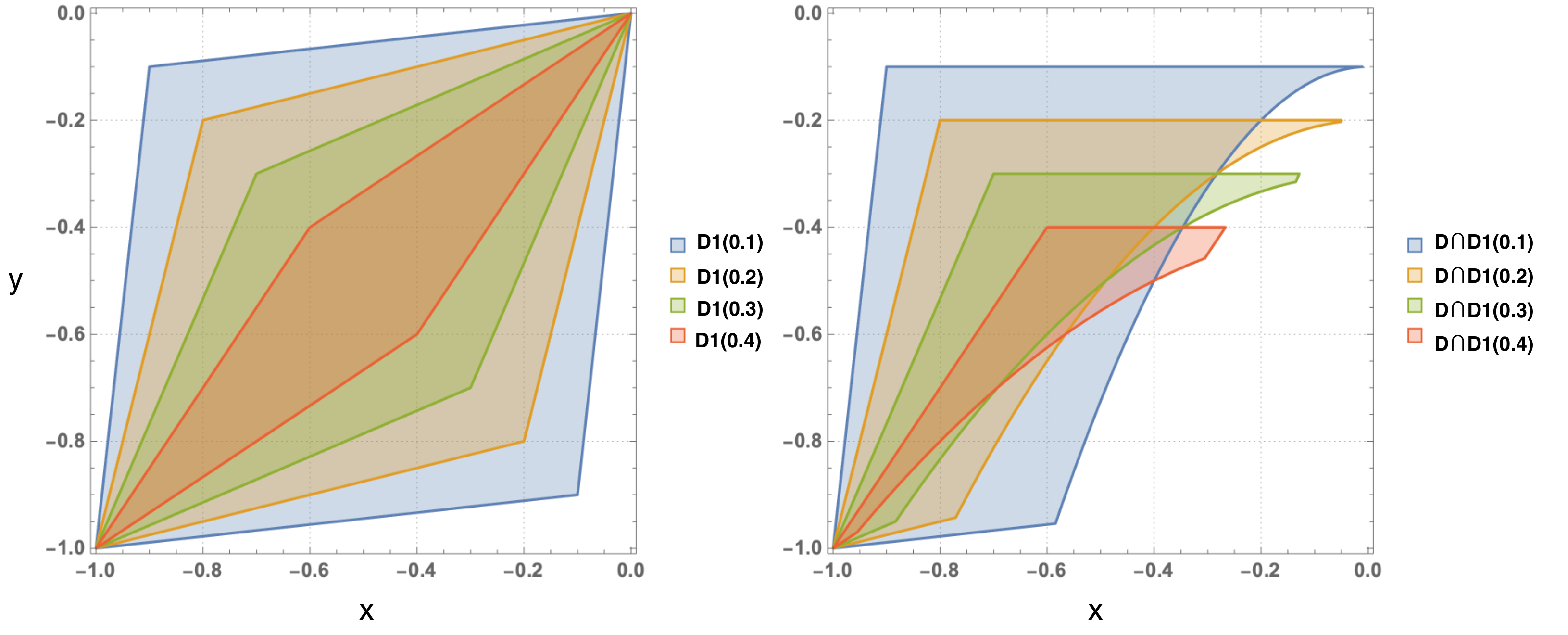}
\caption{SH game: Regions $D_1$ and $D \cap D_1$ are shown for different values of $q$, namely, $q = 0.1, \ 0.2, \ 0.3$ and 0.4.  }
\label{fig: D1}
\end{center}
\end{figure}

\begin{proof} 
In the SH game: $1>T>0>S>-1$, $T\sim U(0,1)$, $S\sim U(-1,0)$. Let $\hat{T}:=T-1\sim U(-1,0)$. The joint distribution of $(\hat{T},S)$ is
$$
f(t,s)=\begin{cases}
1\quad (t,s)\in (-1,0)\times (-1,0),\\
0\quad\text{otherwise}.
\end{cases}
$$
According to \eqref{eq: joint pdf of XY}, the joint probability distribution of $(X,Y)$ is 
\begin{equation*}
g(x,y)=|J|f(\hT(x,y),S(x,y))=\frac{1}{(1-2q)}1_{(x,y)\in D_1},
\end{equation*}
where 
\begin{align*}
D_1&=(\hT(x,y),S(x,y)\in (-1,0)\times (-1,0)
\\&=((x,y)\in (-1,0)^2:~(1-q)x-qy,(1-q)y-qx)\in (-(1-2q),0)\times (-(1-2q),0)
\\&=\{(x,y)\in (-1,0)^2:~-(1-2q)<(1-q)x-qy<0, -(1-2q)<(1-q)y-qx<0 \}
\\&=\{(x,y)\in (-1,0)^2:~\frac{(1-q)x}{q}<y<\frac{(1-q)x+(1-2q)}{q}, \frac{qx-(1-2q)}{1-q}<y<\frac{qx}{1-q} \}.
\end{align*}
We can characterise $D_1$ further by explicitly ordering the lower and upper bounds in the above formula. We have
\begin{align*}
&\frac{(1-q)x}{q}-\frac{qx}{1-q}=\frac{(1-2q)x}{q(1-q)}<0\\
& \frac{(1-q)x}{q}-\frac{qx-(1-2q)}{1-q}=\frac{(1-2q)(x+q)}{q(1-q)},
\\ & \frac{(1-q)x+(1-2q)}{q}-\frac{qx}{1-q}=\frac{(1-2q)(x+1-q)}{q(1-q)}.
\end{align*}
It follows that: 
\begin{enumerate}[(i)]
\item for $-1<x<-(1-q)<-q$:
$$
\frac{(1-q)x}{q}<\frac{qx-(1-2q)}{1-q}<\frac{(1-q)x+(1-2q)}{q}<\frac{qx}{1-q}
$$
\item For $-(1-q)<x<-q$
$$
\frac{(1-q)x}{q}<\frac{qx-(1-2q)}{1-q}<\frac{qx}{1-q}<\frac{(1-q)x+(1-2q)}{q}.
$$
\item for $-q<x<0$
$$
\frac{qx-(1-2q)}{1-q}<\frac{(1-q)x}{q}<\frac{qx}{1-q}<\frac{(1-q)x+(1-2q)}{q}.
$$
\end{enumerate}
Hence 
\begin{align*}
D_1&=\{-1<x<-(1-q),~ \frac{qx-(1-2q)}{1-q}<y<\frac{(1-q)x+(1-2q)}{q}\}
\\&\qquad~\cup~ \{-(1-q)<x<-q,~ \frac{qx-(1-2q)}{1-q}<y<\frac{qx}{1-q}\}
\\&\qquad~\cup~\{-q<x<-0,~ \frac{(1-q)x}{q}<y<\frac{qx}{1-q}\}.
\end{align*}
Thus, $D_1$ is the quadrilateral $ABOC$ with vertices
$$
A=(-1,-1), \quad B=(-(1-q), -q), \quad C=(-q,-(1-q))\quad O=(0,0).
$$
\end{proof}
\begin{theorem}
\label{thm: SH}
 The probability that SH games have $3$ equilibria is given by
\begin{equation*}
p_3^{SH}=
\begin{cases}
1-\frac{q}{2(1-q)}-\frac{1}{1-2q}\Big[\frac{(3\sqrt{q}+2)^2\sqrt{q}(5q^{3/2}+3q^2-9q-3\sqrt{q}+4)}{12(\sqrt{q}+1)^3}
\\\hspace*{4cm}+\frac{-27 q^3-18q^2-32\sqrt{1-2q}q+48q+16\sqrt{1-2q}-16}{12q}\Big],\quad 0<q\leq 4/9\\ \\
1-\frac{q}{2(1-q)}-\frac{8\sqrt{q}(1-2q)^2}{3(1-q)^3},\quad 4/9<q<0.5.
\end{cases}
\end{equation*}
\end{theorem}
\begin{proof}
According to Proposition \ref{prop: p3areaSH}, in order  to compute $p_3^{SH}$ we need to compute the area of $D\cap D_1$.  To this end we need to understand the intersections of the parabola $y=-\frac{x^2}{4q}-q$ with the lines $y=\frac{(1-q)x}{q}$ and $y=\frac{qx-(1-2q)}{1-q}$. The parabola $y=-\frac{x^2}{4q}-q$ always intersects with the line $y=\frac{(1-q)x}{q}$ inside the domain $(-1,0)^2$ at the point
$$
F=(2(q+\sqrt{1-2q}-1), 2(1-q)(q+\sqrt{1-2q}-1)/q).
$$
By comparing $2(q+\sqrt{1-2q}-1)$ with $-q$, it follows that the point $F$ is inside the edge $OC$ if $0\leq q\leq 4/9$ and is outside $OC$ whenever $4/9<q<0.5$. On the other hand, the parabola $y=-\frac{x^2}{4q}-q$ meets the line  $y=\frac{qx-(1-2q)}{1-q}$ at two points  $G_1=(x_1, -\frac{x_1^2}{4q}-q)$ and $G_2=(x_2, -\frac{x_2^2}{4q}-q)$ with
$$
x_1=2q \Big(-\frac{q}{1-q}-\frac{2q-1}{(q-1)\sqrt{q}}\Big),\quad x_2=2q\Big(\frac{2q-1}{(q-1)\sqrt{q}}-\frac{q}{1-q}\Big).
$$
By comparing $x_1$ and $x_2$ with $-q$ we have $G_1$ is always in the edge $AC$, while $G_2$ is outside $AC$ if $0< q\leq 4/9$ and is inside $AC$ if $4/9<q<0.5$.

In conclusion
\begin{enumerate}
\item \textbf{for $0<q\leq 4/9$}: the intersection $D\cap D_1$ is the domain formed by vertices $A$, $B$, $E$, $F$, and $G_1$ where $E=(-\frac{q^2}{1-q}, -q)$ (which is the intersection of $y=-q$ with $y=\frac{(1-q)x}{q}$) and
\begin{itemize}
\item $A$ and $B$ are connected by the line $y=\frac{(1-q)x+(1-2q)}{q}$,
\item $B$ and $E$ are connected by the line $y=-q$,
\item $E$ and $F$ are connected by the line $y=\frac{(1-q)x}{q}$,
\item $F$ and $G_1$ are connected by the parabola $y=-\frac{x^2}{4q}-q$,
\item $G_1$ and $A$ are connected by the line $y=\frac{qx-(1-2q)}{1-q}$.
\end{itemize}
\item \textbf{for $4/9<q<0.5$}: the intersection $D\cap D_1$ is the domain formed by the vertices $A, B, E, C, G_1, G_2$ where
\begin{itemize}
\item $A$ and $B$ are connected by the line $y=\frac{(1-q)x+(1-2q)}{q}$,
\item $B$ and $E$ are connected by the line $y=-q$,
\item $E$ and $C$ are connected by the line $y=\frac{(1-q)x}{q}$,
\item $C$ and $G_1$ are connected by the line $y=\frac{qx-(1-2q)}{1-q}$,
\item $G_1$ and $G_2$ are connected by the parabola $y=-\frac{x^2}{4q}-q$,
\item $G_2$ and $A$ are connected by the line $y=\frac{qx-(1-2q)}{1-q}$.
\end{itemize}
\end{enumerate}
See Figure \ref{fig: SH} for illustrations of the two cases above, with $q=0.3$ and $q=0.45$ respectively.
\begin{figure}
\begin{center}
\includegraphics[width=\linewidth]{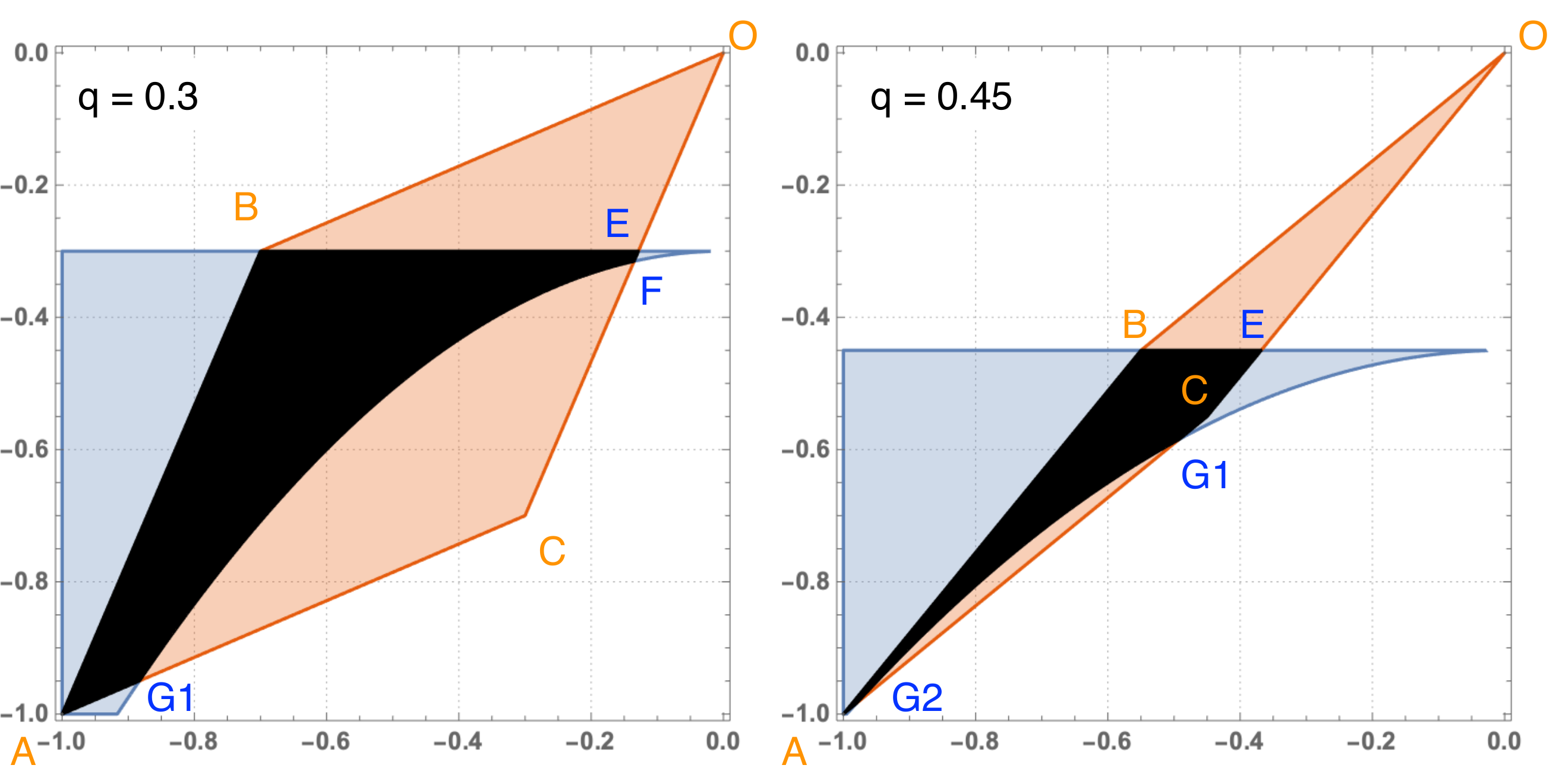}
\caption{Details of $D\cap D_1$ for $q=0.3$ and $q=0.45$}
\label{fig: SH}
\end{center}
\end{figure}
We are now ready to compute the probability that the SH game has three equilibria.

\textbf{For $0<q\leq 4/9$}: The probability that the SH game has three equilibria is
\begin{align}
p^{SH}_3&=\frac{1}{(1-2q)}\mathrm{Area}(D\cap D_1)\notag
\\&=\frac{1}{(1-2q)}\Big[\mathrm{Area}(D_1)-\mathrm{Area}(BOE)-\mathrm{Area}(CFG_1)\Big]\label{area0}
\end{align}
We proceed by computing the areas in the  expression above. Area of $D_1$ is
\begin{align}
\mathrm{Area}(D_1)&=\int_{-1}^{-(1-q)}\Big[\frac{(1-q)x+(1-2q)}{q}-\frac{qx-(1-2q)}{1-q}\Big]\,dx+\int_{-(1-q)}^{-q}\Big[\frac{qx}{1-q}-\frac{qx-(1-2q)}{1-q}\Big]\,dx\notag
\\&\qquad+\int_{-q}^0\Big[\frac{qx}{1-q}-\frac{(1-q)x}{q}\Big]\,dx\notag
\\&=\frac{q(1-2q)}{2(1-q)}+\frac{(1-2q)^2}{1-q}+\frac{q(1-2q)}{2(1-q)}\notag
\\&=1-2q.\label{area1}
\end{align}
Area of $BOE$ is
\begin{align}
\label{area2}
\mathrm{Area}(BOE)=\mathrm{Area}(BOH)-\mathrm{Area}(HOE)=\frac{1}{2}q\Big[(1-q)-\frac{q^2}{1-q}\Big]=\frac{q(1-2q)}{2(1-q)},
\end{align}
where $H=(0,-q)$. Area of $CFG$ is
\begin{align}
\mathrm{Area}(CFG_1)&=\int_{a}^{-q}\Big[\frac{qx-(1-2q)}{1-q}+\frac{x^2}{4q}+q\Big]\,dx+\int_{-q}^{2(q+\sqrt{1-2q}-1)}\Big[\frac{(1-q)x}{q}+\frac{x^2}{4q}+q\Big]\,dx\notag
\\&=\frac{(3\sqrt{q}+2)^2\sqrt{q}(5q^{3/2}+3q^2-9q-3\sqrt{q}+4)}{12(\sqrt{q}+1)^3}\notag
\\&\qquad+\frac{-27 q^3-18q^2-32\sqrt{1-2q}q+48q+16\sqrt{1-2q}-16}{12q}.\label{area3}
\end{align}
Substituting \eqref{area1}, \eqref{area2}, and \eqref{area3} back to \eqref{area0} we obtain, for $0<q\leq 4/9$
\begin{align*}
p_3^{SH}&=1-\frac{q}{2(1-q)}
\\&\qquad-\frac{1}{1-2q}\Bigg[\frac{(3\sqrt{q}+2)^2\sqrt{q}(5q^{3/2}+3q^2-9q-3\sqrt{q}+4)}{12(\sqrt{q}+1)^3}
\\&\qquad\qquad+\frac{-27 q^3-18q^2-32\sqrt{1-2q}q+48q+16\sqrt{1-2q}-16}{12q}\Bigg]
\end{align*}
Now we consider the remaining case $4/9<q<0.5$. In this case
\begin{align}
p^{SH}_3&=\frac{1}{(1-2q)}\mathrm{Area}(D\cap D_1)\notag
\\&=\frac{1}{(1-2q)}\Big[\mathrm{Area}(D_1)-\mathrm{Area}(BOE)-\mathrm{Area}(\lozenge G_1G_2)\Big]\label{area02},
\end{align}
where $\lozenge G_1G_2$ is the domain with vertices $G_1$ and $G_2$ formed by the parabola $y=-\frac{x^2}{4q}-q$ and the line $y=\frac{qx-(1-2q)}{1-q}$. Thus
\begin{align}
\label{area4}
\mathrm{Area}(\lozenge G_1G_2)&=\int_{x_1}^{x_2}\Big[-\frac{x^2}{4q}-q-\frac{qx-(1-2q)}{1-q}\Big]\,dx\notag
\\&=\frac{8\sqrt{q}(1-2q)^3}{3(1-q)^3}.
\end{align}
Substituting \eqref{area1}, \eqref{area2} and \eqref{area4} back to \eqref{area02} we obtain, for $4/9<q<0.5$,
\begin{equation*}
p_3^{SH}=1-\frac{q}{2(1-q)}-\frac{8\sqrt{q}(1-2q)^2}{3(1-q)^3}.
\end{equation*}
This finishes the proof of this theorem
\end{proof}
\subsection{Prisoner's Dilemma (PD)}
\label{sec: PD}
\begin{proposition}
\label{prop: p3PDarea}
The probability that PD games have $3$ equilibria is given by
\begin{equation}
p^{PD}_3=\frac{1}{(1-2q)}\mathrm{Area}(D\cap D_2),
\end{equation}
where $D$ is defined above (as in the case of the SH games) and $D_2$ is the triangle \textit{MNO} with vertices
$$
M=(-q,-(1-q)),\quad N=(0,-\frac{1-2q}{1-q}), \quad O=(0,0).
$$
\end{proposition}
See Figure \ref{fig: D2} for illustration of $D_2$ and its intersection with $D$ for several values of $q$.
\begin{figure}
\begin{center}
\includegraphics[width=\linewidth]{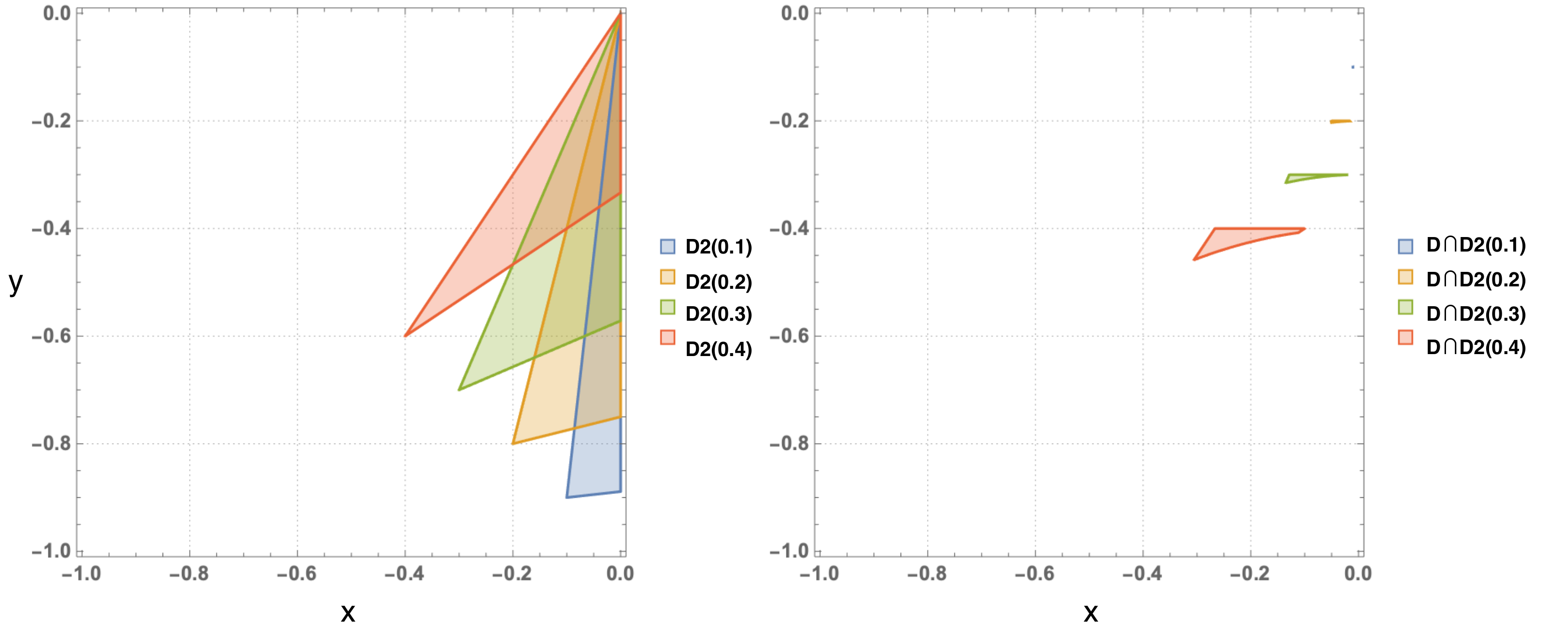}
\caption{PD game: Regions $D_2$ and $D \cap D_2$  are shown  for different values of $q$, namely, $q = 0.1, \ 0.2, \ 0.3$ and 0.4.  }
\label{fig: D2}
\end{center}
\end{figure}

\begin{proof}
Recall that in PD games we have $T\sim U(1,2)$, $S\sim (-1,0)$. Thus $\hT=T-1\sim (0,1)$. The joint distribution of $(\hT, S)$ is
$$
f(t,s)=1_{(x,y)\in (0,1)\times (-1,0)}=\begin{cases}
1\quad (t,s)\in (0,1)\times (-1,0),\\
0\quad \text{otherwise}.
\end{cases}.
$$
According to \eqref{eq: joint pdf of XY}, the joint probability distribution of $(X,Y)$ is 
\begin{equation*}
g(x,y)=|J|f(\hT(x,y),S(x,y))=\frac{1}{(1-2q)}1_{(x,y)\in D_2},
\end{equation*}
where 
\begin{align*}
D_2&=(\hT(x,y),S(x,y)\in (0,1)\times (-1,0)
\\&=\{(x,y)\in (-1,0)^2:~(1-q)x-qy,(1-q)y-qx)\in (0,1-2q)\times (-(1-2q),0)\}
\\&=\{(x,y)\in (-1,0)^2:~0<(1-q)x-qy<1-2q, -(1-2q)<(1-q)y-qx<0 \}
\\&=\{(x,y)\in (-1,0)^2:\frac{(1-q)x-(1-2q)}{q} <y<\frac{(1-q)x}{q}, \frac{qx-(1-2q)}{1-q}<y<\frac{qx}{1-q} \}.
\end{align*}
We now characterise $D_2$ further. We have
\begin{enumerate}[(1)]
\item for $-1<x<-q$ then
$$
\frac{(1-q)x-(1-2q)}{q}<\frac{(1-q)x}{q}<\frac{qx-(1-2q)}{1-q}<\frac{qx}{1-q}
$$
\item for $-q<x<0$ then
$$
\frac{(1-q)x-(1-2q)}{q}<\frac{qx-(1-2q)}{1-q}<\frac{(1-q)x}{q}<\frac{qx}{1-q}
$$
\end{enumerate}
It follows that
$$
D_2=\{(x,y):~-q<x<0,~ \frac{qx-(1-2q)}{1-q}<y<\frac{(1-q)x}{q}\}.
$$
Thus $D_2$ is the triangle  \textit{MNO} with vertices
$$
M=(-q,-(1-q)),\quad N=(0,-\frac{1-2q}{1-q}), \quad O=(0,0).
$$
The probability that the SH game has three equilibria is thus 
\begin{align*}
p^{SH}_3&=\int_{D}g(x,y)\,dxdy\nonumber
\\&=\frac{1}{(1-2q)}\int_{D}1_{(x,y)\in D_2}\,dxdy\nonumber\\
&=\frac{1}{(1-2q)}\mathrm{Area}(D\cap D_2).
\end{align*}
This finishes the proof of this proposition.
\end{proof}
The following elementary lemma provides an upper bound for $p_3^{PD}$, particularly implying that it tends to $0$ as $q$ goes to $0$.
\begin{lemma}
\begin{equation}
p_3^{PD}\leq \frac{q}{2(1-q)}.
\end{equation}
As a consequence, $p_3^{PD}$ is always smaller or equal to $0.5$ and tends to $0$ as $q$ tends to 0.
\end{lemma}
\begin{proof}
Area of $D$
\begin{align*}
\mathrm{Area}(D)=\int_{-1}^0\Big[-q-\Big(-\frac{x^2}{4q}-q\Big)\Big]\,dx=\frac{1}{12 q}.
\end{align*}
Area of $D_2$
$$
\mathrm{Area}(D_2)=\int_{-q}^0\Big[\frac{(1-q)x}{q}-\frac{qx-(1-2q)}{1-q}\Big]\,dx=\frac{q(1-2q)}{2(1-q)}.
$$
Hence
\begin{align*}
p_3^{PD}&\leq \frac{1}{1-2q}\min\Big\{\mathrm{Area}(D),\mathrm{Area}(D_2) \Big\}
\\&=\frac{1}{1-2q}\min\Big\{\frac{1}{12 q},\frac{q(1-2q)}{2(1-q)}\Big\}=\frac{q}{2(1-q)}.
\end{align*}
\end{proof}
\begin{theorem}
\label{thm: PD}
The probability that PD games has three equilibria is given by
\begin{equation*}
p_3^{PD}=\begin{cases}
\frac{1}{1-2q}\Big[-\frac{2(q^3+3q^2+(4\sqrt{1-2q}-6)q-2\sqrt{1-2q}+2)}{3q}-\frac{1}{2}\frac{q^3}{(1-q)}\Big],\quad 0\leq q\leq \frac{3-\sqrt{5}}{2},\\
\frac{-16 \left(\sqrt{1-2 q}-1\right) q^{3/2}+2 q^{5/2}+15 q^3+\left(8 \sqrt{1-2 q}-25\right) q^2+q-8 \sqrt{1-2 q}+2 \sqrt{q} \left(8 \sqrt{1-2 q}-5\right)+5}{6 \left(\sqrt{q}-1\right)^3 \left(q^{3/2}+q\right)},\quad \frac{3-\sqrt{5}}{2}<q\leq 4/9,\\
\frac{1}{2}\frac{(1-2q)^2}{q(1-q)},\quad 4/9<q<0.5.
\end{cases}
\end{equation*}
\end{theorem}
\begin{proof}
According to Proposition \ref{prop: p3PDarea}, in order  to compute $p_3^{PD}$ we need to compute the area of $D\cap D_2$. As in proof of Theorem \ref{thm: SH}, the parabola $y=-\frac{x^2}{4q}-q$ always intersects with the line $y=\frac{(1-q)x}{q}$ inside the domain $(-1,0)^2$ at the point
$$
F=(2(q+\sqrt{1-2q}-1), 2(1-q)(q+\sqrt{1-2q}-1)/q).
$$
The point $F$ is inside the edge $DC$ if $0\leq q\leq 4/9$ and is outside $DC$ whenever $4/9<q<0.5$. 

On the other hand, the parabola $y=-\frac{x^2}{4q}-q$ meets the line  $y=\frac{qx-(1-2q)}{1-q}$ at two points  $G_1=(x_1, -\frac{x_1^2}{4q}-q)$ and $G_2=(x_2, -\frac{x_2^2}{4q}-q)$ with
$$
x_1=2q \Big(-\frac{q}{1-q}-\frac{1-2q}{(1-q)\sqrt{q}}\Big),\quad x_2=2q\Big(-\frac{q}{1-q}+\frac{1-2q}{(1-q)\sqrt{q}}\Big).
$$
$G_1$ is always outside edge $AC$. $G_2$ is inside it if $\frac{3-\sqrt{5}}{2}< q\leq 4/9$ and  outside it otherwise. Therefore, we have three cases (see Figure \ref{fig: PD} for illustration).
\begin{figure}
\begin{center}
\includegraphics[width=\linewidth]{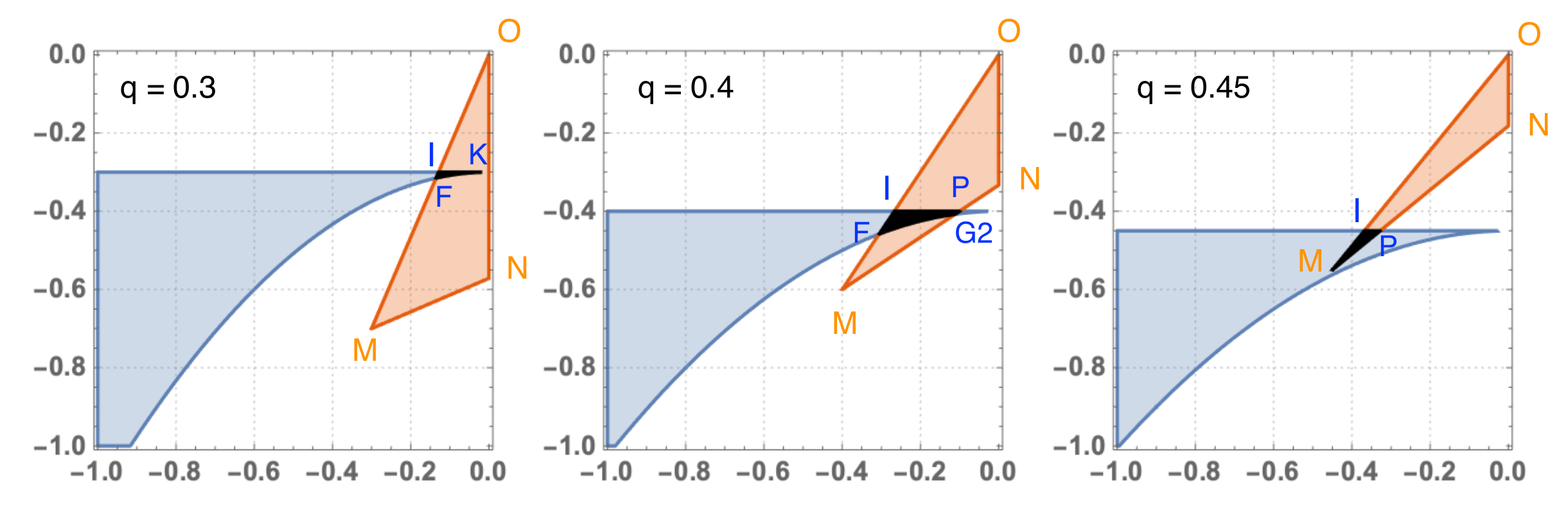}
\caption{$D \cap D_2$ for $q = 0.3, \ 0.4, \ 0.45$}
\label{fig: PD}
\end{center}
\end{figure}

\begin{enumerate}[(1)]
\item \textbf{For $0\leq q\leq \frac{3-\sqrt{5}}{2}$}: the intersection $D\cap D_2$ is formed by $I, F, K$ where $I=(-q^2/(1-q),-q)$ (which is the intersection of $y=-q$ with $y=\frac{(1-q)x}{q}$),  $K=(0.-q)$ and
\begin{itemize}
\item $I$ and $F$ are connected by the line $y=\frac{(1-q)x}{q}$,
\item $F$ and $K$ are connected by the parabola $y=-\frac{x^2}{4q}-q$,
\item $K$ and $I$ are connected by the line $y=-q$.
\end{itemize}
In this case
\begin{align*}
p_3^{PD}&=\frac{1}{1-2q}\mathrm{Area}(D\cap D_2)
\\&=\frac{1}{1-2q}\Big(\mathrm{Area}(OFK)-\mathrm{Area}(OIK)\Big)
\\&=\frac{1}{1-2q}\Bigg(\int_{2(q+\sqrt{1-2q}-1)}^0\Big[\frac{(1-q)x}{q}+\frac{x^2}{4q}+q\Big]\,dx-\frac{1}{2}\frac{q^3}{(1-q)}\Bigg)
\\&=\frac{1}{1-2q}\Big[-\frac{2(q^3+3q^2+(4\sqrt{1-2q}-6)q-2\sqrt{1-2q}+2)}{3q}-\frac{1}{2}\frac{q^3}{(1-q)}\Big].
\end{align*}
\item \textbf{For $\frac{3-\sqrt{5}}{2}<q\leq 4/9$}: the intersection $D\cap D_2$ is formed by $F$, $I$, $P=((q^3-3q+1)/q,-q)$ (which is the intersection of $y=-q$ with $y=\frac{qx-(1-2q)}{1-q}$) and $G_2$, where 
\begin{itemize}
\item $F$ and $I$ are connected by the line $y=\frac{(1-q)x}{q}$,
\item $I$ and $P$ are connected by the line $y=-q$,
\item $P$ and $G_2$ are connected by the line $y=\frac{qx-(1-2q)}{1-q}$,
\item $G_2$ and $F$ are connected by the parabola $y=-\frac{x^2}{4q}-q$.
\end{itemize}
In this case
\begin{align}
p_3^{PD}&=\frac{1}{1-2q}\mathrm{Area}(D\cap D_2)\notag
\\&=\frac{1}{1-2q}\Bigg(\int_{2(q+\sqrt{1-2q}-1)}^{-\frac{q^2}{1-q}}\Big[\frac{(1-q)x}{q}+\frac{x^2}{4q}+q\Big]\,dx+\int_{-\frac{q^2}{1-q}}^{x_2}\frac{x^2}{4q}\,dx+\int_{x_2}^{(q^2-3q+1)/q}\Big[-q-\frac{qx-(1-2q)}{1-q}\Big]\, dx\Bigg)\notag
\\&=\frac{-16 \left(\sqrt{1-2 q}-1\right) q^{3/2}+2 q^{5/2}+15 q^3+\left(8 \sqrt{1-2 q}-25\right) q^2+q-8 \sqrt{1-2 q}+2 \sqrt{q} \left(8 \sqrt{1-2 q}-5\right)+5}{6 \left(\sqrt{q}-1\right)^3 \left(q^{3/2}+q\right)}.
\end{align}
\item \textbf{For $4/9<q<0.5$}: the intersection $D\cap D_2$ is the triangle \textit{MIP} where $M=(-q,-(1-q))$
\begin{itemize}
\item $M$ and $I$ are connected by the line $y=\frac{(1-q)x}{q}$,
\item $I$ and $P$ are connected by the line $y=-q$,
\item $P$ and $M$ are connected by the line $y=\frac{qx-(1-2q)}{1-q}$.
\end{itemize}
In this case
\begin{align*}
p_3^{PD}&=\frac{1}{1-2q}\mathrm{Area}(D\cap D_2)
\\&=\frac{1}{1-2q}\mathrm{Area}(MIP)
\\&=\frac{1}{2(1-2q)}\Big(\frac{q^2-3q+1}{q}+\frac{q^2}{1-q}\Big)\Big(1-2q\Big)
\\&=\frac{1}{2}\frac{(1-2q)^2}{q(1-q)}.
\end{align*}
\end{enumerate}
In conclusion
\begin{equation*}
p_3^{PD}=\begin{cases}
\frac{1}{1-2q}\Big[-\frac{2(q^3+3q^2+(4\sqrt{1-2q}-6)q-2\sqrt{1-2q}+2)}{3q}-\frac{1}{2}\frac{q^3}{(1-q)}\Big],\quad 0\leq q\leq \frac{3-\sqrt{5}}{2},\\
\frac{-16 \left(\sqrt{1-2 q}-1\right) q^{3/2}+2 q^{5/2}+15 q^3+\left(8 \sqrt{1-2 q}-25\right) q^2+q-8 \sqrt{1-2 q}+2 \sqrt{q} \left(8 \sqrt{1-2 q}-5\right)+5}{6 \left(\sqrt{q}-1\right)^3 \left(q^{3/2}+q\right)},\quad \frac{3-\sqrt{5}}{2}<q\leq 4/9,\\
\frac{1}{2}\frac{(1-2q)^2}{q(1-q)},\quad 4/9<q<0.5.
\end{cases}
\end{equation*}
This completes the proof of this theorem.
\end{proof}
\section{Summary and outlook}
\label{sec: summary}
It has been shown that in human behaviours and other biological settings, mutation is non-negligible~\cite{traulsen:2009aa,rand2013evolution,zisis2015generosity}. How mutation affects the complexity and bio-diversity of the evolutionary systems is a fundamental question in evolutionary dynamics \cite{nowak:2006bo,sigmund:2010bo}. In this paper, we have addressed this question for random social dilemmas by computing explicitly the probability distributions of the number of equilibria in term of the mutation probability. Our analysis based on random games is highly relevant and practical,  because  it is often the case that one  might  know the nature of a game at hand (e.g., a coordination or cooperation dilemma), but it is  very difficult and/or costly to measure the exact values  of the game's  payoff matrix.  Our results have clearly shown the influence of the mutation on the number of equilibria in SH-games and PD-games. The probability distributions in these games are much more complicated than in  SD-games and H-games and significantly depend on the mutation strength. For a summary of our results, see again Box 1 and Figure \ref{fig: plot all games}. Our analysis has made use of suitable changes of variables, which expressed the probability densities in terms of area of certain domains. For future work, we plan to generalise our method to other approaches to studying  random social dilemmas such as finite population dynamics and payoff disturbances \cite{huang:2010aa, szolnoki2019seasonal,Amaral2020a, Amaral2020b}, as well as  to multi-player social dilemma games \cite{Pacheco2009,Souza2009, Luo2021}.

\section*{Acknowledgments} TAH is also supported by Leverhulme Research Fellowship (RF-2020-603/9).

\bibliographystyle{plain}

\end{document}